\documentclass[12pt]{article}
\usepackage{amsmath}
\usepackage{amssymb}
\usepackage{amsthm}
\usepackage{caption}
\usepackage{graphicx}
\usepackage{enumerate}
\usepackage{natbib}
\usepackage{subfig}
\usepackage{url} 

\newcommand{\blind}{1}

\begin{document}

\def\spacingset#1{\renewcommand{\baselinestretch}%
{#1}\small\normalsize} \spacingset{1}


\if1\blind
{
  \title{\bf A Proper Concordance Index for Time-Varying Risk}
  \author{A. Gandy\\
    Department of Mathematics, Imperial College London, SW7 2AZ, U.K.\\ 
    a.gandy@imperial.ac.uk \\
    and \\
    T. J. Matcham \hspace{.2cm}\\
    Department of Mathematics, Imperial College London, SW7 2AZ, U.K.\\
    NIHR ARC Northwest London, SW10 9NH, U.K.\\
thomas.matcham14@imperial.ac.uk
    }
  \maketitle
} \fi

\if0\blind
{
  \bigskip
  \bigskip
  \bigskip
  \begin{center}
    {\LARGE\bf A Proper Concordance Index for Time-Varying Risk}
\end{center}
  \medskip
} \fi

\bigskip
\begin{abstract}

Harrel's concordance index is a commonly used discrimination metric for survival models, particularly for  models where the relative ordering of the risk of individuals is time-independent, such as the proportional hazards model. There are several suggestions, but no consensus, on how it could be extended to models where relative risk can vary over time, e.g.\ in case of crossing hazard rates. We show that these concordance indices are not proper, in the sense that they are maximised in the limit by the true data generating model. Furthermore, we show that a concordance index is proper if and only if the risk score used is concordant with the hazard rate at the first event time for each comparable pair of events. Thus, we suggest using the  hazard rate as the time-varying risk score when calculating concordance. Through simulations,  we demonstrate situations in which other concordance indices can lead to incorrect models being selected over a true model, justifying the use of our suggested risk prediction in both model selection and in loss functions in, e.g., deep learning models.

\end{abstract}

\noindent%
\vfill

\newpage
\spacingset{1.9} 
\section{Introduction}
\label{sec:intro}

Accurate patient prognosis estimation is an important clinical tool with applications including advising patients of their likely disease outcomes, informed selection of patient treatment as well as the design and evaluation of clinical trials. There exist several approaches to quantifying the predictive accuracy of survival models \citep{harrell1996multivariable}, which can in turn be optimized for, in order to improve a given aspect of the predictions. Discrimination metrics focus on a model's ability to correctly order the predictions of the patient outcomes. This could  be important, for example, in deciding the order in which a set of patients should be treated. 

The most significant metric of survival model discrimination is Harrel's concordance index \citep{Harrel:1984}, hereafter the C-index, which was first developed as an adaptation of the Kendall-Goodman-Kruskal-Somers type rank correlation index \citep{Goodman:1954} to right-censored survival data, similar to an adaptation of Kendall's $\tau$ by \citet{brown1973nonparametric} and \citet{schemper1984analyses}. 

We use the following setup. 
Let $(X_i, U_i, Z_i), i\in \mathbb{N}$, be independent and identically distributed with the lifetime $X_i$ and the right censoring time $U_i$ being non-negative random variables. Let the covariate $Z_i$ be an element of some  space $\mathcal{Z}$. We observe  $(T_i, D_i, Z_i),i=1,\dots,n$, where  $T_i=\min(X_i,U_i)$ is the time at risk  and $D_i=\mathbb{I}(X_i\leq U_i)$ is the event indicator.

The C-index estimates the probability that the predicted risk scores of a pair of individuals is concordant with that of their observed survival times. Only for pairs of individuals $(i,j)$, with $i\neq j$, for whom the first event is not a censoring event is an ordering of the outcome possible, i.e., the pair is comparable. The probability that individual $i$  has such an event occurring before event $j$ is $P(D_i=1,T_i<T_j)$. In this work we instead use
\begin{equation*}
    \pi_{comp} = P(D_i=1,T_i\leq T_j).
\end{equation*}
This has no effect on the continuous time case as $P(T_i=T_j)=0$, however it is critical for the discrete case in Section \ref{sec3}, where ties will be possible. 

Survival predictions are differentiated with functions of the covariate called risk scores. In situations where the relative risk of individual is not changing over time, e.g., in a proportional hazards model with only time-constant covariates, this is sufficient to discriminate between individuals. However, the risk score should arguably be time-dependent in situations where the relative risk of individuals changes over time, e.g., in cases where hazards of risk groups cross \citep{Mantel1988-qh}, where a proportional hazards model has time-dependent covariates, or  where the risk prediction is individual over time as in machine learning approaches to survival models \citep{Lee:2018}. Two situations in which we see crossing occur include when surgery or more aggressive medication incur a high initial hazard before eventually reducing the overall risk relative to the control group \citep{james2017abiraterone, rothwell1999prediction}.




Thus, the risk score  we use is allowed to depend on the covariate and on time. Specifically,  in a paired comparison, we compare the risk scores at the time when the first event occurs. Intuitively, this comparison gives a prediction of who was most at imminent risk of the event, given that they have survived until the first event time. This framework covers previous specific suggestions for dealing with time-varying risks \citep{Antolini:2005,blanche2019c,haider2020effective}.
Formally, the risk score is specified through a function $q:[0,\infty)\times {\cal Z}\to \mathbb{R}$ and for a given pair $(i,j)$, with $T_i\leq T_j$, we say that $i$ has a higher risk score than $j$ if  $q(T_i|Z_i)>q(T_i|Z_j)$. Higher values of the risk score indicate a propensity towards earlier events. 

For a pair $(i,j)$, where we observe that $i$ has occurred before j, i.e., $D_i=1,T_i\leq T_j$ we say that this pair is concordant if $q(T_i|Z_i)>q(T_i|Z_j)$. The probability of a pair having  observed the event of $i$ before the event of $j$ and  being concordant is 
\begin{equation}
\label{piconc_notiedpred}
 P(D_i=1, T_i\leq T_j, q(T_i|Z_i)>q(T_i|Z_j)).
\end{equation}
Defining a concordance index as (\ref{piconc_notiedpred}) divided by $\pi_{comp}$, would imply the following:
a perfect model that could correctly order every pair would have a concordance of 1,  a model that simply guesses for each pair would have a concordance of 0.5 on average, and a model that always orders incorrectly would have a concordance of 0. A model that gives the same prediction for each individual would also only get a concordance of 0, which seems undesirable.  

To avoid the latter, tied risk scores are often rewarded with a score of 0.5, such that a  model with  the same risk score for everyone would still score 0.5 \citep{harrell1996multivariable}. Hence,  in the 
C-index  $$C_q = \frac{\pi_{conc}}{\pi_{comp}}$$ we use
\begin{equation*}
\pi_{conc} = P[D_i=1,T_i\leq T_j,q(T_i|Z_i)>q(T_i|Z_j)] +  \frac{1}{2}P[D_i=1,T_i\leq T_j,q(T_i|Z_i)=q(T_i|Z_j)].
\end{equation*}

Given a random sample  $(T_i,D_i,Z_i)^n_{i=1}$ we can estimate $\pi_{conc}$ and $\pi_{comp}$ with:
\begin{align*}
\hat{\pi}_{conc} = \frac{1}{n(n-1)}\sum^n_{i=1}\sum^n_{j=1;j\neq i}&\{\mathbb{I}[D_i=1, T_i\leq T_j, q(T_i|Z_i)>q(T_i|Z_j)]\\&+\frac{1}{2}\mathbb{I}[D_i=1,T_i\leq T_j,q(T_i|Z_i)=q(T_i|Z_j)]\},\\
\hat{\pi}_{comp} = \frac{1}{n(n-1)}\sum^n_{i=1}\sum^n_{j=1;j\neq i}&\mathbb{I}(D_i=1,T_i\leq T_j)
\end{align*}
and thus estimate the C-index $C_q$ by  $$c ^n_q = \hat{\pi}_{conc}/ \hat{\pi}_{comp}.$$

Often, the risk score $q(t|z)$ being used is not dependent on the first argument $t$. For example, if a proportional hazards model with covariates $z$ is used, then often 
 the linear predictor $q(t|z) = z\hat\beta$ is used as risk score, where $\hat \beta$ is an estimate of the regression coefficient.

For more general survival models, where we have access to a survival function $S(t|z)$ as a function of the covariates $z$, a definition of a risk score is less obvious, as there may not be a clear definition of what constitutes higher risk, for example when the underlying hazard rates of individuals cross. Several methods of computing risk scores in this setting have been considered, for example 
$q(t|Z) = -S(t_0|Z)$, the negative of the survival function evaluated at some fixed time $t_0>0$ \citep{blanche2019c}, or
$q(t|Z)= -\inf\{t \text{   s.t   } S(t|Z) \leq 0.5\}$, the negative of the median survival time \citep{haider2020effective}. The negative is taken as predicted survival times have the opposite ordering to risk scores. Again, these suggestions do not depend on the first argument of $q$.

In the work of \citet{Antolini:2005}, a time-dependent concordance index, $C^{td}$,  is introduced. This adaptation of the C-index is developed for models with either time-varying covariates or time-varying effects, while supposing the predicted survival function is the 'natural' relative risk predictor. This leads to an event-time dependent risk score
\begin{equation*}    
q(t|z) = -S(t|z).
\end{equation*}
This index  is used widely in deep learning survival models, wherein the survival curves for distinct individuals are prone to crossing \citep{Zhong:2021}.  A similar concordance index has seen use in loss functions for deep survival models \citep{Lee:2018}. 

In Sections \ref{sec4} and \ref{sec5}, we show through examples  that these concordance indices do not always maximally reward correct models, and therefore could lead to selection of inferior predictive models. This is reinforced by the work of \cite{rindt2022survival}, wherein they show that $C^{td}$ and other metrics are not proper scoring rules.

The main contribution of this paper (Section \ref{sec2} and \ref{sec3}) is that  using the conditional hazard rate $\alpha(t|z)$  as a time-dependent risk score for an individual with covariates $z$ does behave analogously to a proper scoring rule. Since the definition of a proper scoring rule cannot be directly applied to the concordance index, we define a proper concordance index as $C_q$ defined above with a risk score $q$ such that
$$\forall \tilde{q}:[0,\infty)\times\mathcal{Z}\rightarrow\mathbb{R}: \quad  C_q\geq C_{\tilde{q}}.$$ Thus we suggest using $q(t|z)=\alpha(t|z)$ as risk score in concordance indices. Finally, we demonstrate in Section \ref{sec6} the advantage of using this risk score when training deep learning models, both as an element of the loss function, as well as in model validation.

\section{Continuous Event Time}
\label{sec2}

The following theorem shows that, the estimated concordance index $c_q^n$ converges in probability to the concordance  $C_q$ for any risk score 
$q:[0,\infty)\times\mathcal{Z}\to\mathbb{R}$ and  that the concordance is maximised iff the risk score is concordant with the hazard rate.

We assume  that  $X_i$ and $U_i$ are independent given $Z_i$, i.e., $X_i\perp \!\!\! \perp U_i\mid Z_i$,  that $X_i|Z_i$ has an absolutely continuous distribution, and  that there exists $\alpha:[0,\infty)\times\mathcal{Z}\rightarrow[0,\infty)$ such that the hazard rate of $X_i$ given $Z_i$ is $\alpha(t| Z_i)$. We also assume that $U_i\leq {\cal T}$ for some ${\cal T}\in \mathbb{R}$, i.e.\ that we have a finite observation window.

\newtheorem{theorem}{THEOREM}
\begin{theorem} \label{thm:1}
Under the continuous time set-up, if   $\pi_{comp}>0$ 
then
$$
c_q^n\overset{p}{\to} C_q \quad (n\to \infty).
$$
Furthermore, the following equivalence holds:\\
$C_q$ is a proper concordance index, i.e.,
$$\forall \tilde{q}:[0,\infty)\times\mathcal{Z}\rightarrow\mathbb{R}: \quad  C_q\geq C_{\tilde{q}}$$
if and only if for $i\neq j$:

\begin{equation}
    \label{eq:riskhazconcordant}
    \begin{split}
    E\int_0^{\tau_{ij}}
    \{&\mathbb{I}[q(s\vert Z_i)\geq q(s\vert Z_j), \alpha(s\vert Z_i)<\alpha(s\vert Z_j)] \\
    &+
    \mathbb{I}[q(s\vert Z_i)\leq q(s\vert Z_j), \alpha(s\vert Z_i)>\alpha(s\vert Z_j)]
    \}ds=0.
    \end{split}
\end{equation}
where $\tau_{ij}=T_i \wedge T_j$.
\end{theorem}

Equation (\ref{eq:riskhazconcordant}) is trivially satisfied if $q=\alpha$, which is why we suggest using the  hazard rate as the risk score.
More generally, (\ref{eq:riskhazconcordant}) is satisfied if the risk score $q$ and the hazard rate $\alpha$ are concordant in the sense that $\forall s\in[0,\infty), z_1,z_2\in\mathcal{Z}: q(s|z_1)>q(s|z_2)\iff\alpha(s| z_1)>\alpha(s| z_2)$.

To show Theorem \ref{thm:1},  we need to introduce some  counting process notation.
$$  N_{ij}^{comp}(t)= \mathbb{I}(T_i \leq t, D_i=1,T_i\leq T_j) $$
indicates if the event for  $i$ is known to have occurred before the event for  $j$ by time $t$.
The counting process $N_{ij}^{conc,1}(t)$ indicates if additionally the risk scores are in line with $i$ occurring before $j$, i.e.,
$$
    N_{ij}^{conc,1}(t) = N_{ij}^{comp}(t)\cdot \mathbb{I}[q(T_i|Z_i)>q(T_i|Z_j)]
$$
and $N_{ij}^{conc,2}(t)$ indicates if additionally the risk scores for $i$ and $j$ are tied, i.e., 
$$
    N_{ij}^{conc,2}(t) = N_{ij}^{comp}(t)\cdot \mathbb{I}[q(T_i|Z_i)=q(T_i|Z_j)].
$$
$N_{ij}^{conc}(t)$ adds these two together, with tied predictions instead contributing 1/2, i.e., 
$$
    N_{ij}^{conc}(t)=N_{ij}^{conc,1}(t) +\frac{1}{2}N_{ij}^{conc,2}(t).
$$
Based on the above, we now  define the concordance of $n$ individuals using information up to time t as
\begin{equation}\label{conc16}
    c^n_q(t) = \frac{\sum^{n}_{i=1}\sum^n_{j=1,j\neq i}N_{ij}^{conc}(t) }{\sum_{i=1}^n\sum_{j=1,j\neq i}^nN_{ij}^{comp}(t) }.
\end{equation}
We have $c_q^n=c_q^n({\cal T})$.

The following lemma derives the compensator of $N_{ij}^{conc}$ with respect to the filtration $({\cal F}_t)_t$, 
where 
$\mathcal{F}_t=\sigma(Z_i, \mathbb{I}(T_i\leq s), \mathbb{I}(T_i\leq s, D_i=1), i\in \mathbb{N}, 0\leq s\leq t)$
is the information observed up to time $t$.

\newtheorem{lemma}{LEMMA}
\begin{lemma} \label{thm:lem1}
$N_{ij}^{conc}(t)$ has a unique decomposition into a martingale  $M_{ij}^{conc}(t)$ and compensator $$\Lambda_{ij}^{conc}(t)=\int_0^tY_{ij}(s)[Q_{ij}^1(s)+ \frac{1}{2}Q_{ij}^2(s)]\alpha(s|Z_i)ds,$$
where 
$     Y_{ij}(t) = \mathbb{I}(\tau_{ij} \geq t)$,  $\tau _{ij}=T_i\wedge T_j$, 
$  Q_{ij}^1(t) = \mathbb{I}[q^{\tau_{ij}}(t| Z_i)>q^{\tau_{ij}}(t| Z_j)]$, 
$     Q_{ij}^2(t) = \mathbb{I}[q^{\tau_{ij}}(t| Z_i)=q^{\tau_{ij}}(t| Z_j)]$, and  
$q^{\tau_{ij}}(t|\cdot)=q(t\wedge \tau_{ij}|\cdot)$.
\end{lemma}
The proof of this lemma can be found in the Appendix.

\begin{proof}[Proof of Theorem \ref{thm:1}]
$\hat \pi_{conc}$ can be written as a U-statistic
$$
    \hat \pi_{conc} = \frac{1}{n(n-1)}\sum_{i=1}^{n-1}\sum_{j=i+1}^nh[(T_i,D_i,Z_i),(T_j,D_j,Z_j)]
$$with the kernel 
\begin{equation*}
    h[(T_i,D_i,Z_i),(T_j,D_j,Z_j)] = N_{ij}^{conc}({\cal T})+N_{ji}^{conc}({\cal T}).
\end{equation*}The kernel $h$ is bounded, implying $Eh^2[(T_i,D_i,Z_i),(T_j,D_j,Z_j)]<\infty$, and thus Theorem 12.3 of 
 \citet{vaart:1998}  shows that $\hat \pi_{conc}$ is asymptotically normal  as $n\rightarrow \infty$ with mean $\frac{1}{2} E[N_{ij}^{conc}(t)+N_{ji}^{conc}(t)]=E[N_{ij}^{conc}({\cal T})]=\pi_{conc}$. 
 Thus, we have $\hat \pi_{conc}\overset{p}{\to}\pi_{conc}$ as $n\to \infty$.
Similarly, we can show $\hat \pi_{comp}\overset{p}{\to} \pi_{comp}$ as $n\to \infty$.
Hence, by the assumption $\pi_{comp}>0$, we have $c_n^q=\hat\pi_{conc}/\hat \pi_{comp}\overset{p}{\to}\pi_{conc}/\pi_{comp}=C_q$ as $n\to \infty$.

Our choice of $q$ has no influence on the denominator, so considering the numerator only we find that, using Lemma \ref{thm:lem1},
\begin{align*}
   2  E[N_{ij}^{conc}(t)] =&E[N_{ij}^{conc}(t)+N_{ji}^{conc}(t)] =E[\Lambda_{ij}^{conc}(t)+\Lambda_{ji}^{conc}(t)] + E[M_{ij}^{conc}(t)+M_{ji}^{conc}(t)]\nonumber\\
    =&E[\Lambda_{ij}^{conc}(t)+\Lambda_{ji}^{conc}(t)] + 0
    = E\int_0^tf_q(s)Y_{ij}(s)ds\nonumber    
    = E\int_0^{t\wedge\tau_{ij}} f_q(s) ds,
\end{align*}
where 
\begin{align*}
    f_q(s)=&\alpha(s|Z_i)\mathbb{I}[q(s|Z_i)>q(s|Z_j)]+\alpha(s|Z_j)\mathbb{I}[q(s|Z_i)<q(s|Z_j)]\\
           & +0.5[\alpha(s|Z_i)+ \alpha(s|Z_j)]\mathbb{I}[q(s|Z_i)=q(s|Z_j)].
\end{align*}

Let $F_q=E\int_0^{\tau_{ij}} f_q(s) ds$ and let 
\begin{align*}
    A_q(s)=&\mathbb{I}[q(s\vert Z_i)\geq q(s\vert Z_j), \alpha(s\vert Z_i)<\alpha(s\vert Z_j)] +\\
    &\mathbb{I}[q(s\vert Z_i)\leq q(s\vert Z_j), \alpha(s\vert Z_i)>\alpha(s\vert Z_j)].
\end{align*}
Then, for any $q$, 
\begin{align*}
F_\alpha-F_q=E\int_0^{\tau_{ij}}[f_\alpha(s)-f_q(s)]ds
=E\int_0^{\tau_{ij}}[f_\alpha(s)-f_q(s)]A_q(s)ds,
\end{align*}
as $f_q(s)=f_\alpha(s)$ if $A_q(s)=0$. The latter can be seen by going through the three cases $\alpha(s\vert Z_i)>\alpha(s\vert Z_j)$, $\alpha(s|Z_i)<\alpha(s|Z_j)$ and $\alpha(s|Z_i)=\alpha(s|Z_j)$. Furthermore,  $A_q(s)=1$ implies $f_\alpha(s)>f_q(s)$. Thus,
$F_\alpha \geq F_q$ and $F_\alpha=F_q$ if and only if $E \int_0^{\tau_{ij}} A_q(s)ds=0$.
\end{proof}

\section{Discrete Event Time}
\label{sec3}
We show that an analogous result to Theorem \ref{thm:1}  holds for discrete time data. To show the result we need to treat  pairs of events with tied event times ($T_i=T_j, i\neq j$) as comparable.

Suppose that the possible event and censoring times $X_i$ and $U_i$ are discrete random variables over the positive integers $\mathbb{N}^+$.  We denote the discrete hazard rate by $\alpha(t| Z_i)=P(X_i=t\vert X_i \geq t, Z_i)$. As before, we assume  $X_i\perp\!\!\!\perp U_i\mid Z_i$ and that there is a finite observation window ensured by $U_i\leq\cal{T}$ for some $\cal{T}\in \mathbb{N}^+$. The definitions of $T_i$, $D_i$, $\pi_{comp}$,  $c_q^n$ and $C_q$ are as in the previous sections and risk scores are now defined as 
$q:\mathbb{N}^+\times\mathcal{Z}\rightarrow[0,1]$.

\begin{theorem} \label{thm:2}
Under the discrete time set-up, if   $\pi_{comp}>0$ 
then
$$
c_q^n\overset{p}{\to} C_q \quad (n\to \infty).
$$
Furthermore, the following equivalence holds:\\
$C_q$ is a proper concordance index, i.e.,
$$\forall \tilde{q}:\mathbb{N}^+\times\mathcal{Z}\rightarrow[0,1]: \quad  C_q\geq C_{\tilde{q}}$$
if and only if for $i\neq j$:

\begin{equation}
    \label{eq:riskhazconcordant_discrete}
    \begin{split}
    E\sum_{s=1}^{\tau_{ij}}
    \{&\mathbb{I}[q(s\vert Z_i)\geq q(s\vert Z_j), \alpha(s\vert Z_i)<\alpha(s\vert Z_j)] \\
    &+
    \mathbb{I}[q(s\vert Z_i)\leq q(s\vert Z_j), \alpha(s\vert Z_i)>\alpha(s\vert Z_j)]
    \}=0.
    \end{split}
\end{equation}
where $\tau_{ij}=T_i \wedge T_j$.
\end{theorem}
The proof of which is similar to that of Theorem \ref{thm:1} and can be found in the Appendix. The decision to treat pairs with tied event times as comparable pulls each concordance score towards 0.5. Also, the scores of different models are pulled closer together, while retaining the same ordering. This is because for such pairs we always have  $N_{ij}^{\text{conc}}(\tau_{ij}) + N_{ji}^{\text{conc}}(\tau_{ij})=1$ and $N_{ij}^{\text{comp}}(\tau_{ij}) + N_{ji}^{\text{comp}}(\tau_{ij})=2$. We prove these statements fully in Appendix \ref{App:Ties}.

\section{Demonstration of incorrect model selection}
\label{sec4}
We now present an experiment to compare concordance indices  produced by different risk scores. The set up is chosen to show that it is possible to favour incorrect models over the true data generating mechanism. Further studies would be needed to show how typical this situation is.

We generate a data set with crossing hazards inspired by the problem discussed by \cite{Mantel1988-qh}. Let a population of 2000 be divided into two groups, with covariate $Z_i = 0$ for those in group 0 and $Z_i = 1$ for those in group 1. The data generating model $M_0$ is specified by the hazard rates
$$ \alpha_{M_0}(t\vert Z_i=0) =  0.5, \quad
\alpha_{M_0}(t\vert Z_i=1) = t$$
There is independent right censoring by  an exponential distribution with rate 0.05 as well as censoring for anyone who survives until $t=1.1$.

Now let there be 3 incorrect models $M_1$, $M_2$, $M_3$, for us to compare to, which are  defined by their hazard rates $\alpha_{M_1}$,
$\alpha_{M_2}$, $\alpha_{M_3}$
as follows: 
$$
\alpha_{M_1}(t\vert Z_i=0) =
      0.5, \quad
\alpha_{M_1}(t\vert Z_i=1) =
    \begin{cases}
      t, & ( t\leq  0.5) \\
      10t, & (0.5<t) 
    \end{cases},
$$
$
\alpha_{M_2}(t\vert Z_i=0) =
      0.25,$ $
\alpha_{M_2}(t\vert Z_i=1) =
      t,
$
$
\alpha_{M_3}(t\vert Z_i=0) =
      0.5,$
$      
\alpha_{M_3}(t\vert Z_i=1) =
      0.5t.
$
The hazard and cumulative hazard rates are shown in in Figure \ref{fig:1}.

We use four different risk scores to calculate concordance indices. Our suggestions of the hazard at time of first event  uses  $q(s|Z)=\alpha(s|Z)$ and  is  denoted by $C_\alpha$.  
Survival at time of first event, the suggestion of \cite{Antolini:2005}, is denoted by $C^{td}$ and uses  $q(t|z)=-S(t|z)$, where $S(t|z)$ is the survivor function at time $t$ for an individual with covariate $z$. 
Survival at fixed times 0.5 and 1.05 are denoted by  $C_{S(0.5)}$ and $C_{S(1.05)}$ and use $q(t|z)=-S(0.5|z)$ and $q(t|z)=-S(1.05|z)$, respectively.  The quantile survival time is denoted by $C_{\mu(s)}$ and uses $q(t|z)=-\inf\{u \text{   s.t   } S(u|z) \geq s\} $.

We generated 100 different data sets and computed the resulting concordance indices as well as, for every concordance index, the frequency with which each  model achieved the highest concordance index. Results are presented in Table \ref{fig:1}.

As anticipated by Theorem \ref{thm:1}, $C_\alpha$ almost always selects the correct model, but is  unable to distinguish between $M_0$ and $M_1$ as both models have risk scores concordant with the hazard rate of the true model $M_0$. The concordance $C^{td}$ consistently selects an incorrect model in this situation. $C_{S(0.5)}$ fails to perform any model selection, giving every model an equal score in every experiment.  $C_{S(1.05)}$ mostly selects an  an incorrect model. Finally, $C_{\mu(0.5)}$ similarly chooses an incorrect model in most iterations, while $C_{\mu(0.75)}$ mostly fails to distinguish between $M_0$ and $M_3$, showing that choosing $\mu(s)$ as the risk score can perform as well as $C_\alpha$, but is dependent on s (the best of which will be unknown). With each iteration there is a small chance that the randomly generated data will result in concordance calculation orderings that do not match the order of the expected concordances. This has resulted in a small number of deviations in model selection from the general trend.


\begin{figure}
\includegraphics[width=\linewidth]{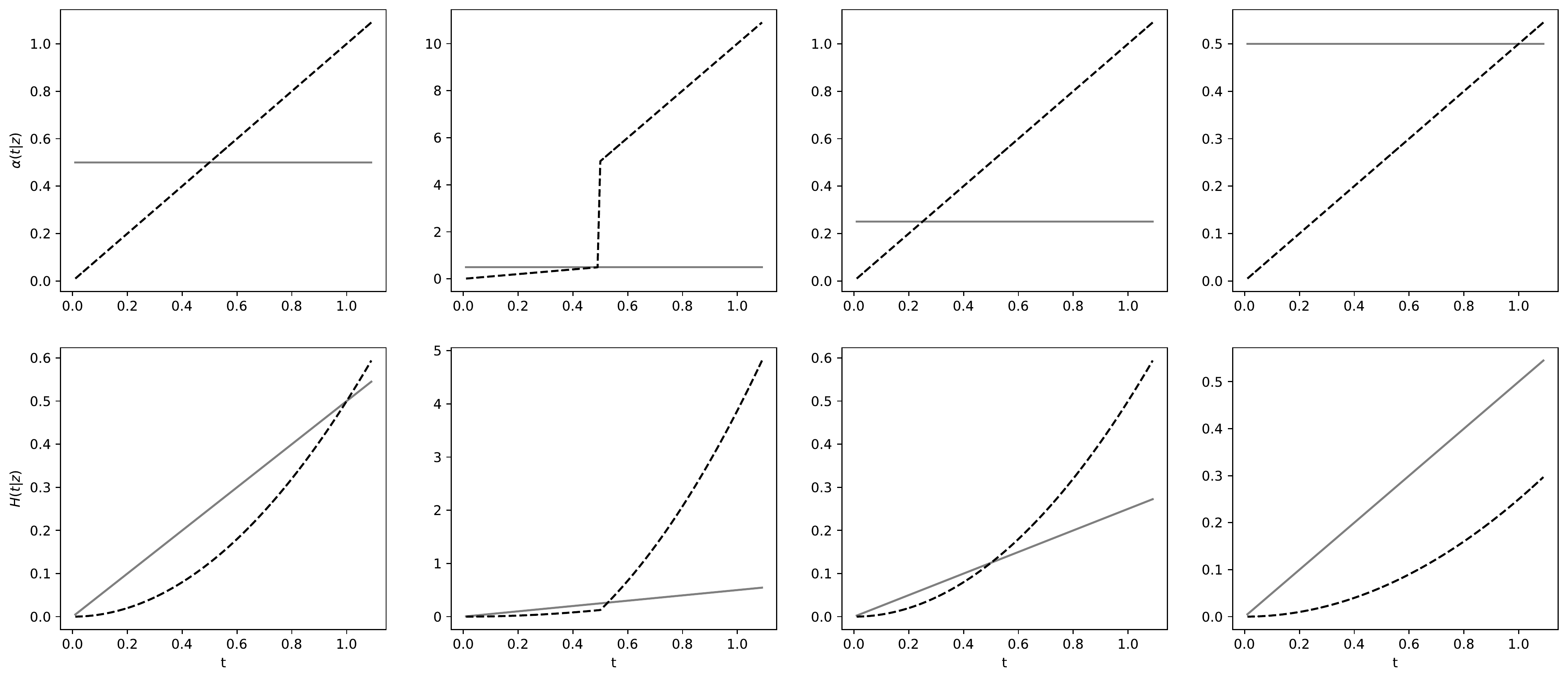}
\caption{Hazard rates (top row) and  cumulative hazard rates (bottom row) of models $M_0, \dots, M_3$ (left to right). Group 0: solid lines; group 1: dotted lines.}
\label{fig:1}
\end{figure}

\begin{table}
    \caption{Simulation from $M_0$ as described in Section 4. Left: Average concordance scores for each model/risk score. Right: Frequency of model selection via the highest risk score 100 replications; tied highest scores counted for all tied models.}
    \label{tab:1}
    \centering
    \begin{tabular}{c c c c c}
    \hline
         & $M_0$ &  $M_1$ & $M_2$ & $M_3$ \\ 
        \hline
          $C_\alpha$ & 0.57 & 0.57 & 0.55 & 0.53  \\
         $C^{td}$ & 0.53 & 0.57 & 0.57 & 0.52   \\
         $C_{S(0.5)}$ & 0.52&  0.52 & 0.52 & 0.52   \\
         $C_{S(1.05)}$ & 0.48 & 0.48 & 0.48 & 0.52  \\
        $C_{\mu(0.5)}$ & 0.48 & 0.48 & 0.48 & 0.52    \\
        $C_{\mu(0.75)}$ & 0.52, & 0.48 & 0.48 & 0.52 \\
        \hline
    \end{tabular}
    \quad
    \begin{tabular}{c c c c c}
    \hline
        & $ M_0$ & $M_1$ & $M_2$ & $M_3$\\
        \hline
        $C_\alpha$&98 & 98 & 2 & 0\\
        $C^{td}$ &0 & 50 & 50 & 0\\
        $C_{S(0.5)}$ &100 & 100 & 100 & 100\\
        $C_{S(1.05)}$ &4 & 4 & 4 & 96\\
        $C_{\mu(0.5)}$&4 & 4 & 4 & 96\\
        $C_{\mu(0.75)}$ & 96 & 4 & 4 & 96\\
        \hline
    \end{tabular}
\end{table}

\section{Comparing Kaplan-Meier Estimates}
\label{sec5}

For a set of right censored survival data, the maximum likelihood estimator over all valid survival distributions is given by the Kaplan-Meier estimator. 
The Kaplan-Meier estimate is useful when simply examining recovery rates and probable event times for groups of individuals, as well as investigating the effectiveness of a treatment. In the latter case, individuals are grouped by treatment, and survival curve estimates for each group are compared using, for example, the log-rank test to establish treatment efficacy. The quality of the fit of such estimates are also commonly evaluated using the C-index. In the following experiment we investigate two further situations, ($M_4, M_5$), with crossing hazard rates between groups, for which Kaplan-Meier estimates will be evaluated by the concordance index using a range of risk scores. 

Let there be two groups of 2000 patients with hazards rates
$$
\alpha_{M_4}(t\vert Z_i=0) =
    \begin{cases}
      6, & ( t\leq  0.1) \\
      1, & (t>0.1)
    \end{cases}, \qquad
\alpha_{M_4}(t\vert Z_i=1) = 1.4
$$
in the first experiment, and as 
$$
\alpha_{M_5}(t\vert Z_i=0) =
    \begin{cases}
      0.5, & ( t\leq  0.9) \\
      10, & (t>0.9)
    \end{cases},\qquad
\alpha_{M_5}(t\vert Z_i=1) =
    \begin{cases}
      2, & ( t\leq  0.9) \\
      1, & ( t> 0.9) 
    \end{cases},
$$
in the second experiment.
The Kaplan-Meier estimates of the survival functions will be calculated for each group, which will then be used to produce risk scores for all event times. The risk score used in $C_\alpha$ requires an estimate of the hazard rate at each event time. One formula for the hazard rate, given a survival function when the survival time has an absolutely continuous distribution function is $\alpha(t) = f(t)/S(t) = \frac{-dS(t)}{dt}/S(t)$ \cite[Example II.4.1]{Andersen:2012}. Therefore, by smoothing and differentiating the Kaplan-Meier estimate of the survival function, we will be able to produce an estimate of the hazard rate. In the following experiment we smooth the Kaplan-Meier survival estimate using a triangular smoothing kernel with bandwidth $b=0.05$. Furthermore, in order to prevent a bias at the beginning and end of the survival curve, we firstly extend $\hat{S}(t)$ by $b$ around t = 0, and secondly we report results right censored $b$ earlier than the true right-censoring time, giving a final right-censor time of 1. For full implementation details refer to Appendix \ref{appendix:code}.

The survival and hazard function estimates are show in Figures \ref{fig:2} and \ref{fig:3}. The hazard rate estimates correctly only cross once in each experiment, near the true hazard crossing times. Concordance results for each risk score are reported in Table \ref{tab:2}. For model $M_4$ we find that all concordance indices, except $C_\alpha$, find that this model has almost no ability to discriminate, scoring only 0.51, whereas $C_\alpha$ scores 0.57, indicating good discrimination. For model $M_5$, $C_\alpha$ finds the model to be even stronger, whereas every other risk score reports the model to be significantly discordant, with scores of 0.44, despite the predicted hazard and survival functions matching the truth well. The estimated survival plot is produced using the lifelines python package \cite{Davidson}, which includes $95\%$ confidence intervals.

This experiment shows a clear shortcoming of calculating the C-index with other risk scores that is not experienced with $C_\alpha$, further justifying its use. Furthermore, $M_4$ may be of special interest as it reflects a treatment setting wherein the treatment group experiences an initial period of higher mortality, followed by a recovered period where they are healthier. We have shown that the other risk scores cannot be relied to recognise strong models in situations where hazard rates cross. These experiments also give a simple method of calculating $C_\alpha$ for the Kaplan-Meier survival function estimate. This method may be less feasible in situations with less data points, as this will result in poorer hazard function estimates.

\begin{figure}%
    \centering
    \subfloat[\centering ]{{\includegraphics[width=7.5cm]{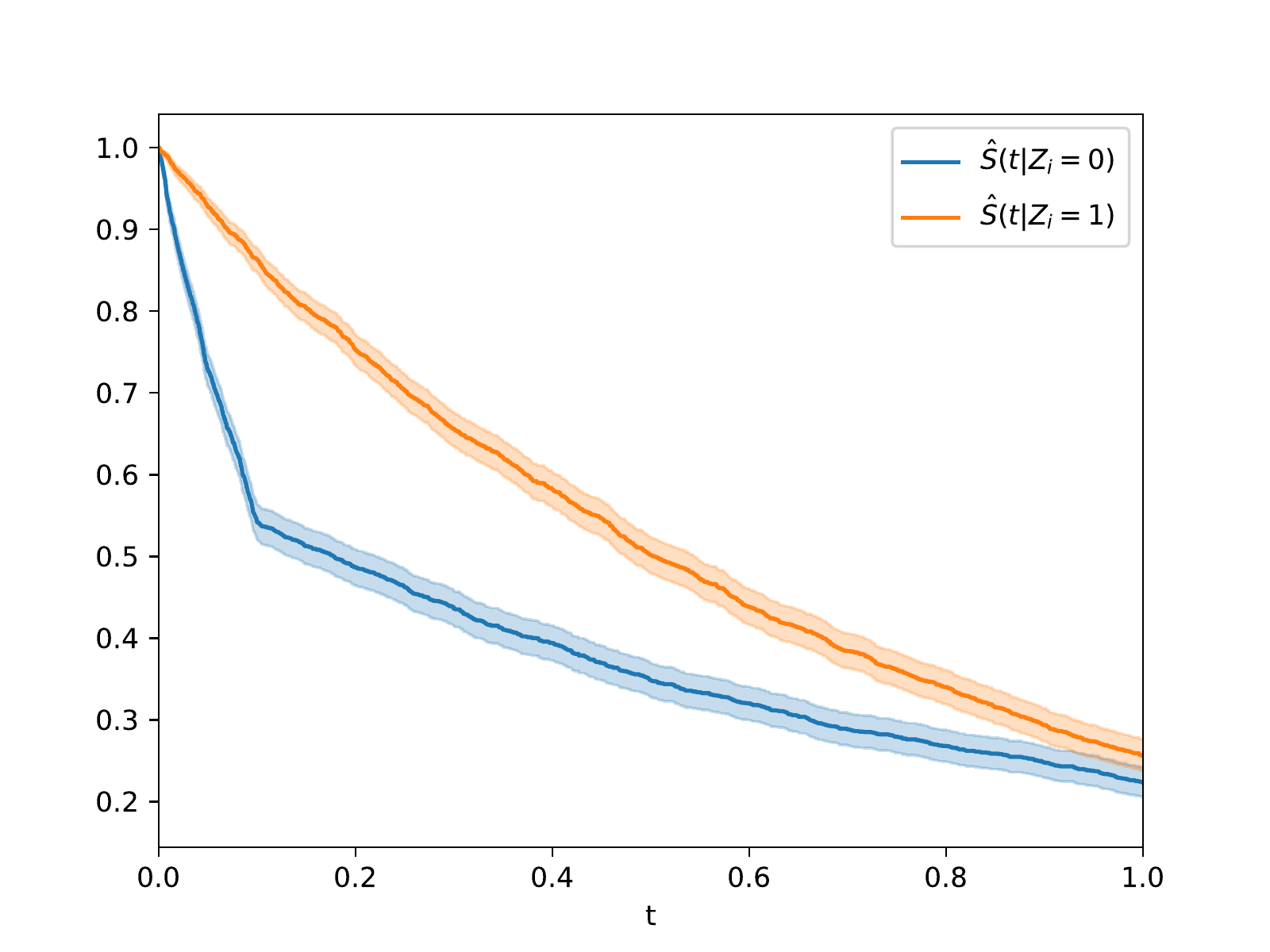} }}%
    \qquad
    \subfloat[\centering ]{{\includegraphics[width=7.5cm]{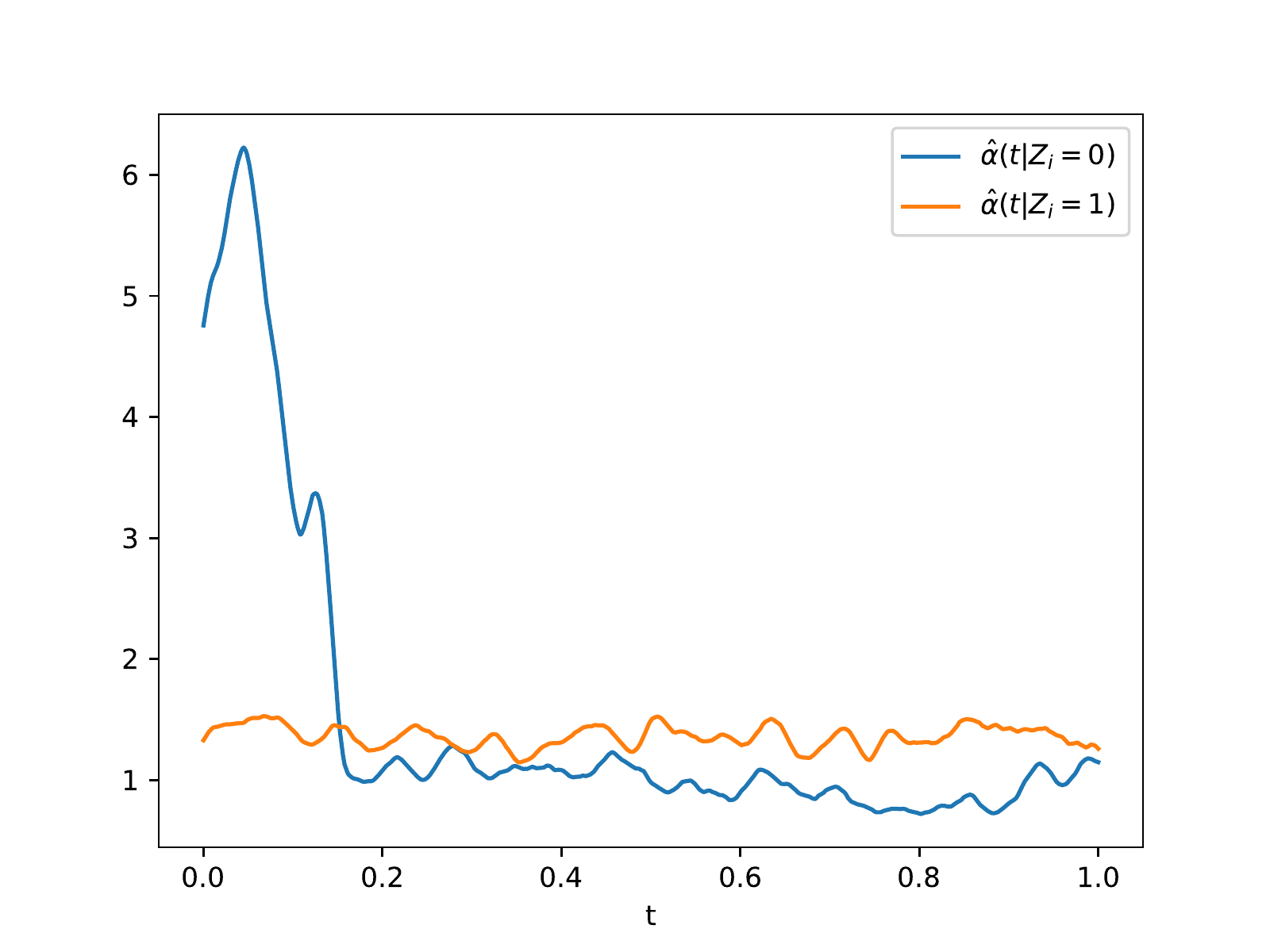} }}%
    \caption{Estimated survival and hazard functions for data from $M_4$.}%
    \label{fig:2}%
\end{figure}

\begin{figure}%
    \centering
    \subfloat[\centering ]{{\includegraphics[width=7.5cm]{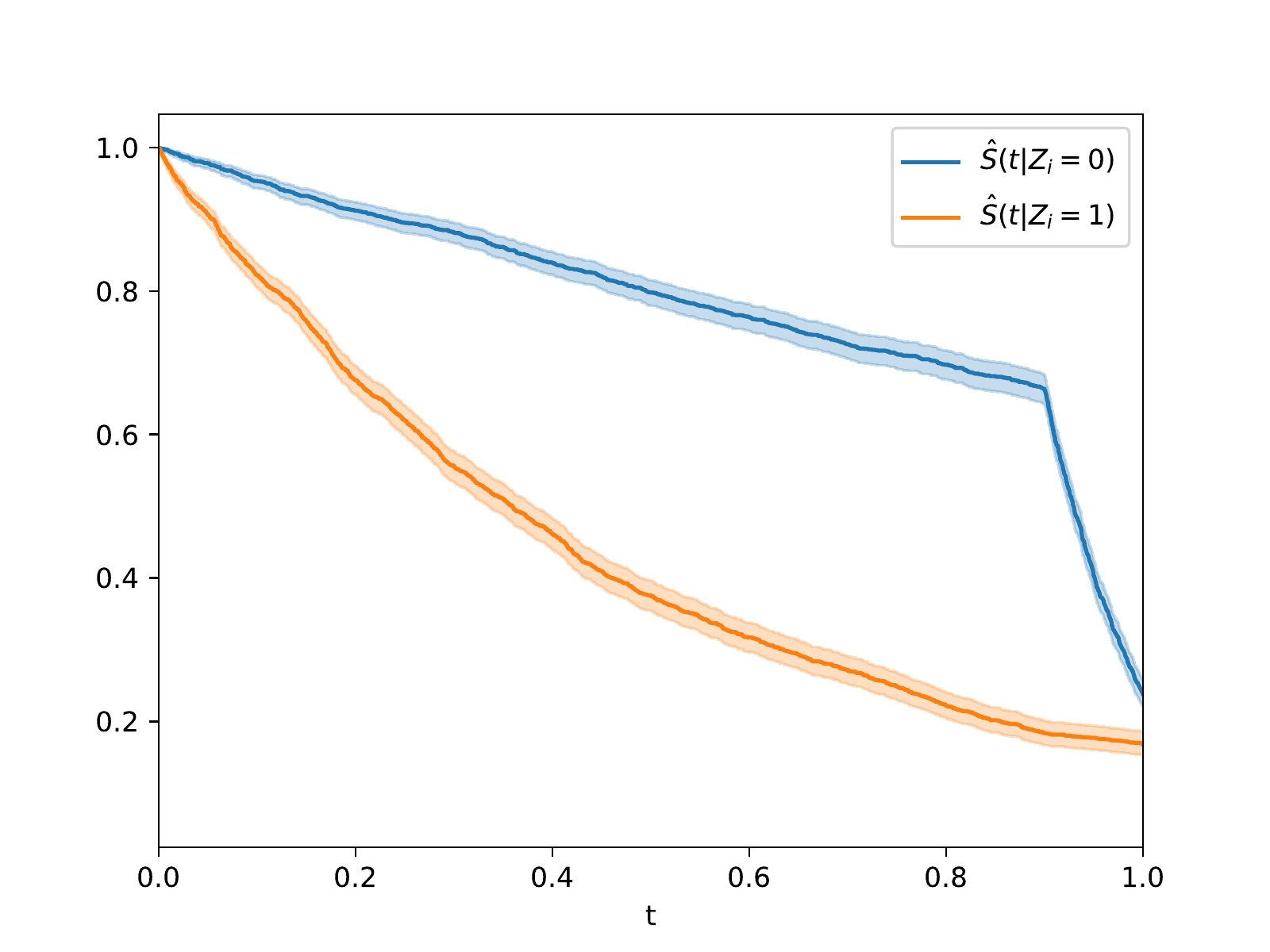} }}%
    \qquad
    \subfloat[\centering ]{{\includegraphics[width=7.5cm]{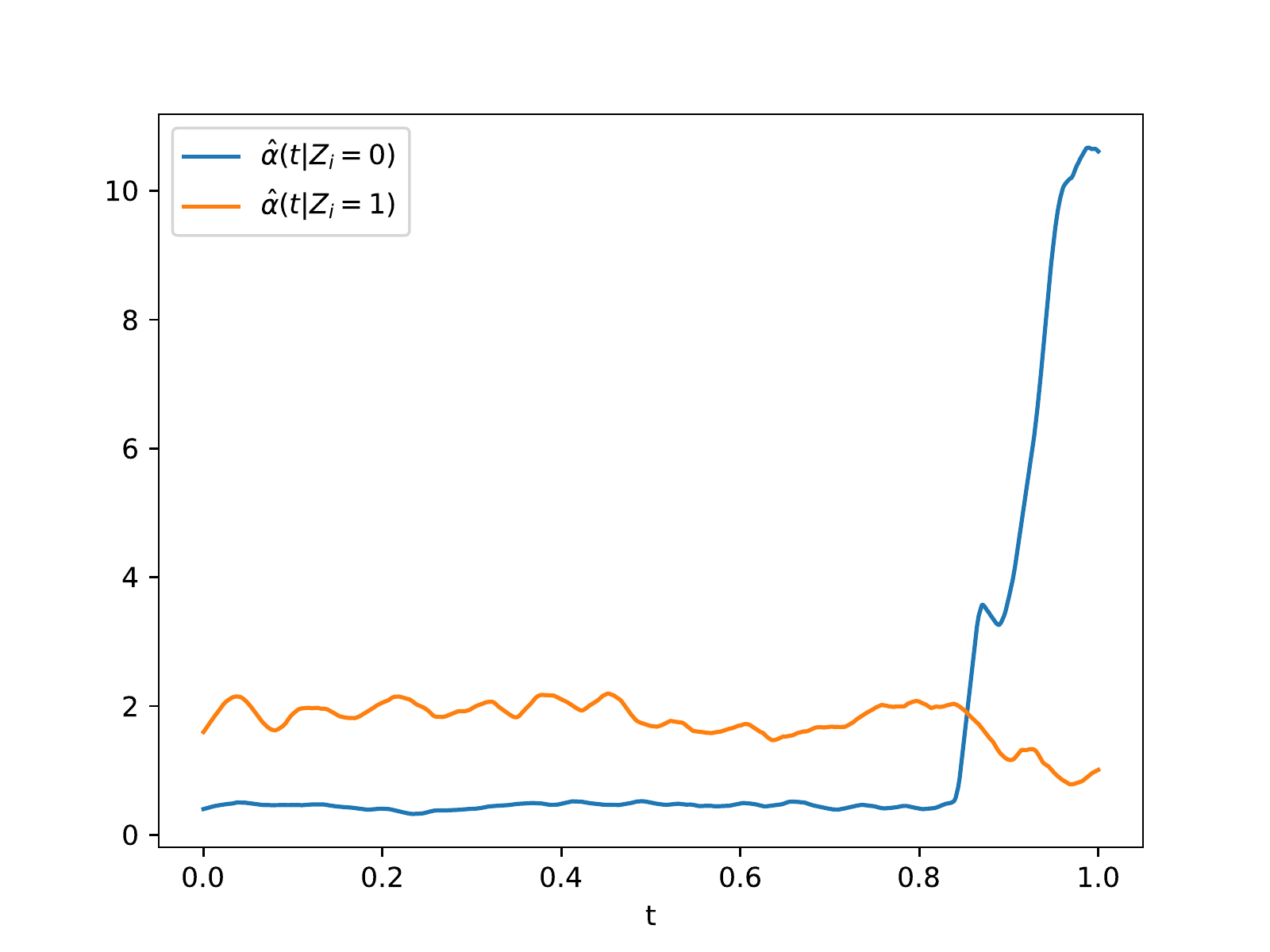} }}%
    \caption{Estimated survival and hazard functions for data generated by $M_5$. }%
    \label{fig:3}%
\end{figure}

\begin{table}
    \caption{ Concordance index scores across range of risk indices for models Kaplan-Meier estimates of models $M_4$ and $M_5$.}
    \label{tab:2}
    \centering
    \begin{tabular}{c c c c c c c}
    \hline
        & $C_\alpha$ & $C^{td}$ & $C_{S(0.5)}$ & $C_{\mu(0.25)}$ & $C_{\mu(0.5)}$ & $C_{\mu(0.75)}$ \\
    \hline
        $M_4$ &0.57 & 0.51 & 0.51 & 0.51 & 0.51 & 0.51 \\
        $M_5$ &0.61 & 0.44 & 0.44 & 0.44 & 0.44 & 0.44 \\
        \hline
    \end{tabular}
\end{table}

\section{Deep Learning}
\label{sec6}
In this section we produce an experiment to test how using $C_{\alpha}$ in the loss function of a deep learning survival model may improve predictions. In the work of \cite{Lee:2018} they present DeepHit, a neural network  that predicts discrete probability mass functions for each $T_i$, $\hat{f}(t\vert Z_i)$. The loss function used to train DeepHit is the sum of two terms, the regular log-likelihood function
\begin{equation}
    L_0 = \sum_{i, D_i=1}\log(\hat{f}(t \vert Z_i))+ \sum_{i, D_i=0}\log(1-\hat{F}(t\vert Z_i)),
\end{equation}
where $\hat{F}(t\vert Z_i) = \sum_{s\leq t}\hat{f}(s\vert Z_i)$. As well as a second term that is designed to encourage the minimisation of $C^{td}$
\begin{equation}
    L^{td} = \sum_{i\neq j} A_{i,j}\cdot \eta (\hat{F}(T_i\vert Z_i), \hat{F}(T_i\vert Z_j)),
\end{equation}
where $A_{i,j}=\mathbf{1}(T_i<T_j, D_i=1)$ indicates which ordering of each pair (i,j) has first experienced the event first (if at all) and $\eta(x,y)=\exp(\frac{-(x-y)}{\sigma})$ for some fixed hyper-parameter $\sigma >0$. It seems this function is chosen instead of the previous concordance equation as it is differentiable and can therefore be optimised for using gradient descent. However, we argue that a loss built to mimic the C-Index should return scores for a pair in as similar a way as possible and  decided to use the sigmoid function $\sigma(x) = \frac{e^{t}}{e^{t}+1}$ instead of the exponential, giving $\eta(x,y)=\sigma(\frac{-(x-y)}{\sigma})$, so that loss incurred by each pair is limited to $[0,1]$ with risk score draws returning 0.5. Another change made was the calculation of C-indices in validation and testing was altered to match that given in Section \ref{sec3}, treating pairs with equal observation time as comparable.

In our experiment we will target $C_\alpha$ by adapting $L^{td}$ to instead evaluate the ordering of the predicted hazard rate $\hat{\alpha}$ at the first event time for each pair.

\begin{equation}
    L_{\alpha} = \sum_{i\neq j} A_{i,j}\cdot \eta (\hat{\alpha}(T_i\vert Z_i), \hat{\alpha}(T_i\vert Z_j))
\end{equation}

To test this model we again produce some synthetic survival data, modelling discrete hazard rates as  DeepHit produces prediction for discrete event times. We let there be two groups of 10,000 patients with  discrete hazard rates

$$
\alpha_{M_6}(t\vert Z_{i,1}=0) =    \begin{cases}
      0.05, & ( t = 1,\dots, 5) \\
      0.5, & (t = 6,\dots, 10)
    \end{cases}, \qquad
\alpha_{M_6}(t\vert Z_{i,1}=1) =
    \begin{cases}
      0.5, & ( t = 1,\dots, 5) \\
      0.05, & ( t = 6,\dots 10) 
    \end{cases},
$$
Noise variables are included alongside the true covariates. For each patient, independent covariates $Z_{i,k}\sim Bernoulli(0.5)$ for $k=2,\dots, 10$ are included. 

We train on 80$\%$ of this data, validate during training on 4$\%$ and test with the remaining 16$\%$. The validation is accomplished by calculating $C^{td}$ for the model trained with the $L^{td}$ loss term and $C_{\alpha}$ for the model trained with $L_{\alpha}$ (since these are the metric each model is targeting)  and training stops if these scores do not improve for 5 training epochs.

\begin{figure}
\centering
\subfloat[]{
  \includegraphics[width=7.5cm]{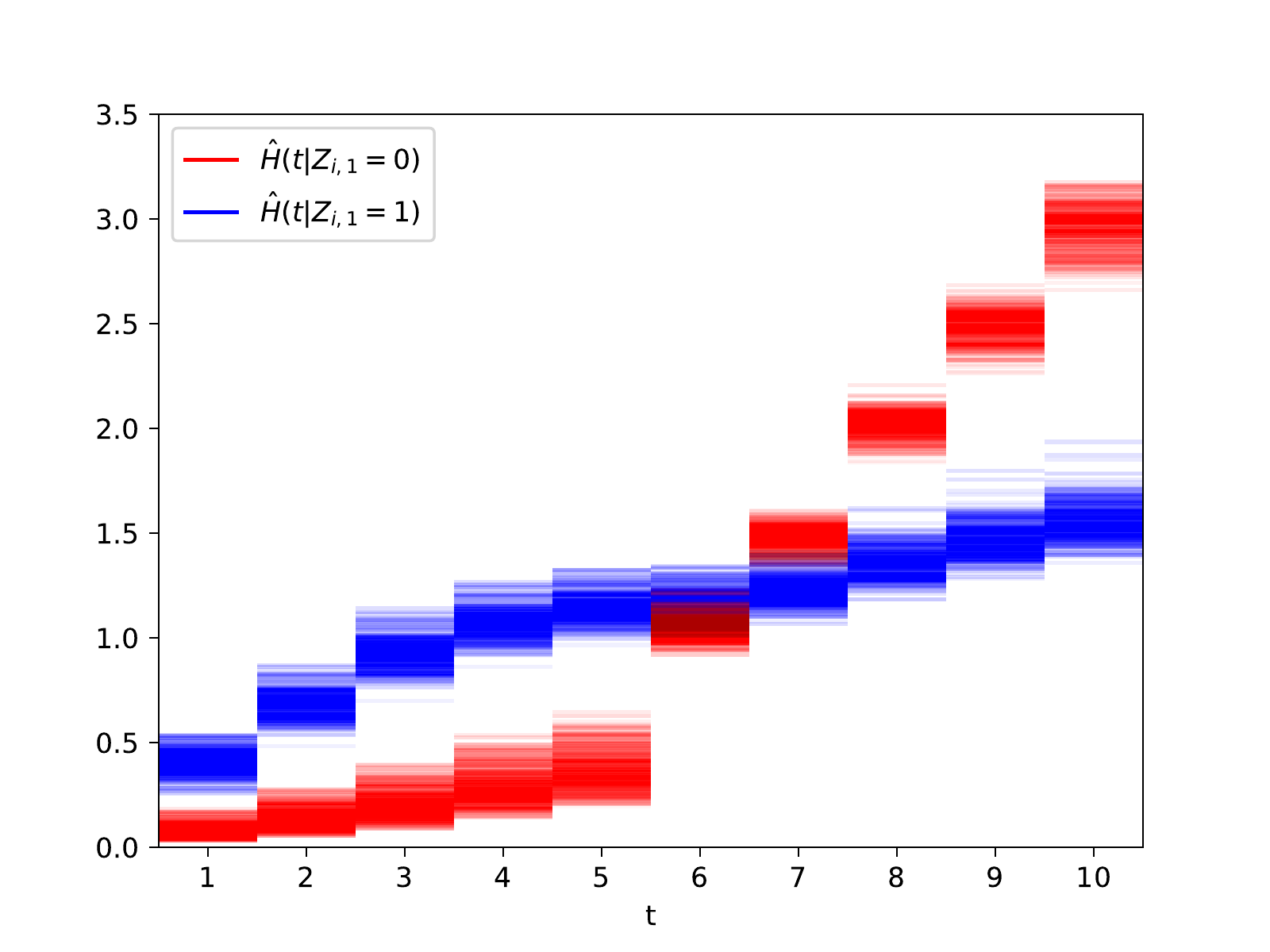}
}
\subfloat[]{
  \includegraphics[width=7.5cm]{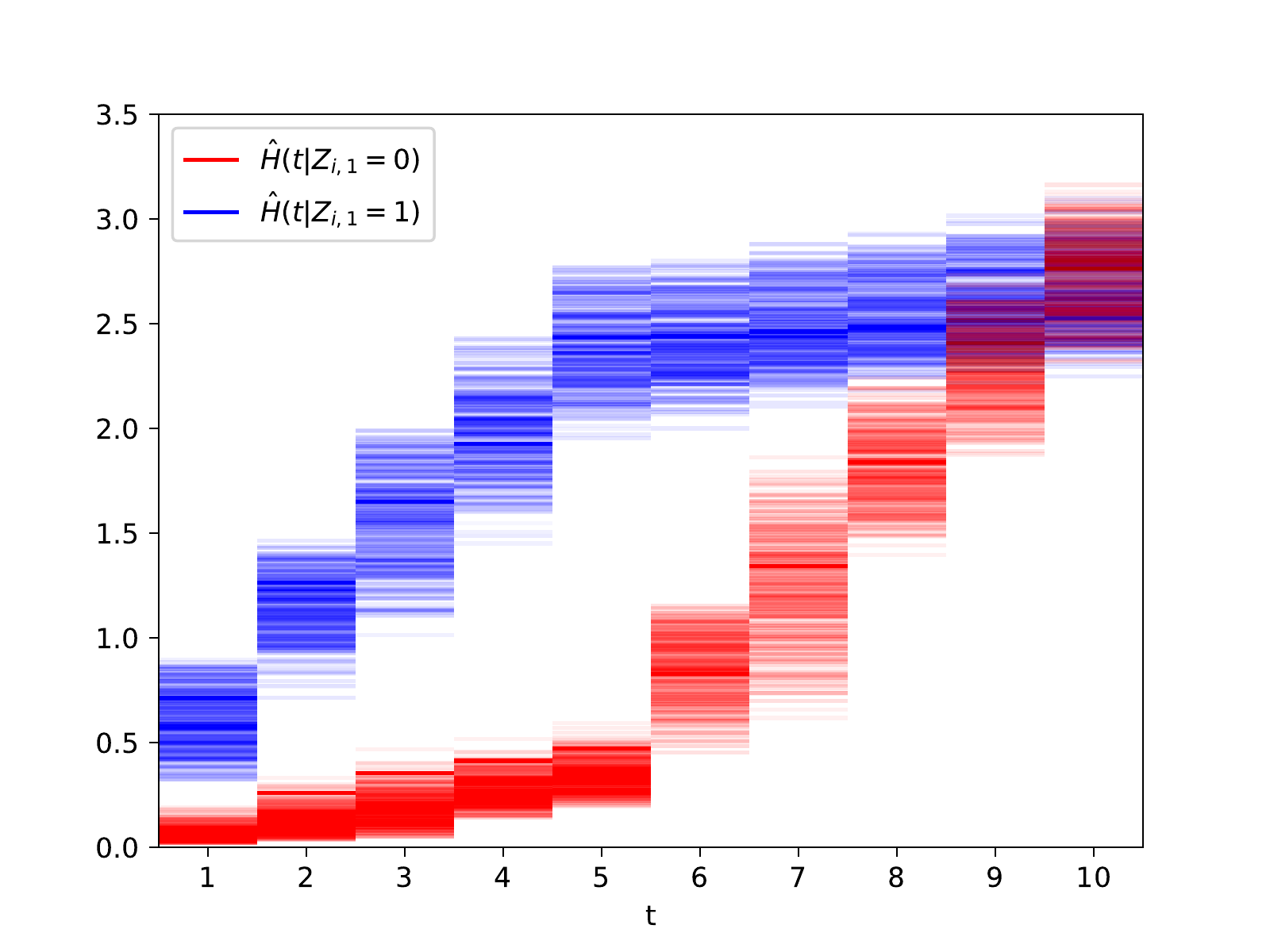}
}
\caption{DeepHit predicted discrete cumulative hazard using $L^{td}$ (a) and $L_{\alpha}$ (b)}
\label{fig:4}%
\end{figure}

The models perform admirably when evaluating them according to the concordance index corresponding to the loss functions used. The model trained with $L_{td}$ loss scores $C_{td}=0.69$, while the model trained with $L_{\alpha}$ loss scores $C_{\alpha}=0.69$ on the testing data. However, comparison of concordance scores in this setting is non-informative as each model is designed to perform well for their respective metrics. Since we know the form of the true hazard rates for each individual we can compare the predictions directly with the truth. In Figure \ref{fig:4} (a) and (b) all discrete cumulative hazard predictions are displayed. 

DeepHit trained with $L_{\alpha}$ performs as expected, predicting the cumulative hazards correctly, with some noise due to the inclusion of the nine noise variables. Conversely, when training with $L^{td}$ the cumulative hazards rates for the second group cross below those of the first group at $t=7$, several steps earlier than it should. 

This error can be explained by the following; suppose the network gave the true hazard rates as the predictions. The loss function based on $C^{td}$ would then evaluate across all pairs of individuals. If we consider the pairs (i,j) such that $Z_{i,1}=0, Z_{j,1} = 1$, $6\leq T_i\leq T_j$ and $D_i=1$. Despite individual i actually having a higher hazard rate at $T_i$, the cumulative hazard for j would be higher, so $L^{td}$ considers the pair non-concordant and would incur a loss. Such spurious losses would then in future training cycles encourage the neural network to re-weight in such a way that the cumulative hazard of j would be lower at $T_i$, while the cumulative hazard for i would be higher.

This experiment is designed to prominently display the consequences of choosing to target $C_{\alpha}$ instead of $C^{td}$ in the loss function of deep learning survival models. Our result shows that if the underlying process generating the data does have crossing hazard rates between individuals, then this trend may not be effectively learned by using $L_{td}$ instead of $L_{\alpha}$. In a real data setting it may be expected that the hazard rate crossing may not be as dramatic, but without knowing the true hazard functions, using $L_{\alpha}$ may be preferred.

\section{Discussion}

In this paper we have explored the limitations of risk scores used in the C-index, specifically when it is used to assess survival models that are capable of producing crossing hazard rate predictions. Previous work has approached solving similar problems, for example Antolini's C-index focused on models with crossing survival curves. In our work we show that even these efforts can still suffer from rewarding incorrect survival models more highly than the truth. The main contribution of this paper is the development of the proper concordance index $C_{\alpha}$, which we prove has the desirable property of asymptotically not scoring a prediction more highly than the true survival distribution. We demonstrated the advantages of $C_{\alpha}$ in two settings, first as the metric of success for Kaplan-Meier models, and secondly its various uses for deep learning models. The deep learning models trained with loss functions incorporating $C_{\alpha}$ outperformed those targeting $C^{td}$ in a situation with crossing hazard rates.  This experiment also further demonstrated the advantages of $C_{\alpha}$ as success metric, as well as it's use as a validation metric during training. 

\newpage
\appendix

\section{Appendix}
\subsection{Proofs}
\label{appendix}

\begin{proof}[Proof of Lemma \ref{thm:lem1}]
Let $i\in \mathbb{N}$.
Consider the counting process $N_{i}(t)=D_i \mathbb{I}(T_i\leq t)$, with a unique decomposition $N_i(t)=\Lambda_i(t)+M_i(t)$ into a  compensator
$\Lambda_i(t)=\int_0^t\alpha(s|Z_i)Y_i(s)ds$ and a finite variation local martingale $M_i$ with respect to $(\mathcal{F}_t)$. 

With $j\in \mathbb{N},$ $j\neq i$, let $N_{i}^{\tau_{ij}}$ be $N_i$ stopped at $\tau_{ij}$, i.e.,  $N_{i}^{\tau_{ij}}(t)=N_{i}(t\wedge \tau_{ij})$.
Since a finite variation local martingale stopped at a stopping time is also a finite variation local martingale $M^{\tau_{ij}}_{i}(t):=M_i(t\wedge \tau_{ij})$ is a finite variation local martingale. By uniqueness of decomposition \citep[Theorem III.16]{Protter:2010}, the compensator of $N^{\tau_{ij}}_{i}(t)$ is therefore 
\begin{align*}
 \Lambda_i^{\tau_{ij}}(t)&
                    = \int_0^{t\wedge \tau_{ij}}\alpha(s|Z_i)Y_i(s)ds
                    = \int_0^t\alpha(s|Z_i)Y_{ij}(s)ds,
\end{align*}
where $Y_{ij}(t) = \mathbb{I}(t \leq \tau_{ij} )$.
Letting  $Q_{ij}^1(t)=\mathbb{I}[q(t\wedge\tau_{ij}|Z_i)>q(t\wedge\tau_{ij}|Z_j)]$, we can write $N_{ij}^{conc,1}$ as
\begin{align*}
    N_{ij}^{conc,1}(t) &= \int_0^t Q_{ij}^1(s)dN_{i}^{\tau_{ij}}(s)= \Lambda_{ij}^{conc,1}(t) + M_{ij}^{conc,1}(t),
\end{align*}
where $\Lambda_{ij}^{conc,1}(t):= \int_0^t Q_{ij}^1(s)d\Lambda^{\tau_{ij}}_{i}(s) $ and $M_{ij}^{conc,1}(t) := \int_0^t Q_{ij}^1(s)dM^{\tau_{ij}}_{i}(s)$.

 $M_{ij}^{conc,1}(t)$ is a local martingale with respect to $(\mathcal{F}_t)$ as $Q_{ij}^1$ is predictable and bounded \cite[Theorem II.3.1]{Andersen:2012}. Thus, again by uniqueness of decomposition, the compensator of $N_{ij}^{conc,1}(t)$ is given by $\Lambda_{ij}^{conc,1}(t)$, which can be rewritten as
\begin{align*}
    \Lambda^{conc,1}_{ij}(t) &= \int_0^t\alpha(s|Z_i)Q^1_{ij}(s)Y_{ij}(s)ds,
\end{align*}
implying that the intensity of $N_{ij}^{conc,1}(t)$ is  
$\lambda_{ij}^{conc,1}(t) = \alpha(t|Z_i)Q_{ij}^1(t)Y_{ij}(t)$. By similar arguments we can show that the
process $N_{ij}^{conc,2}(t)$ has a decomposition into local martingales and compensators with intensity process 
\begin{align*}
    \lambda_{ij}^{conc,2}(t) &= \alpha(t|Z_i)Q_{ij}^2(t)Y_{ij}(t),
\end{align*}
where $Q_{ij}^2(t) = \mathbb{I}[q(t\wedge\tau_{ij}|Z_i)=q(t\wedge\tau_{ij}|Z_j)]$. 
Now we can decompose $N_{ij}^{conc}$ as 
\begin{align*}
    N_{ij}^{conc}(t)&=N_{ij}^{conc,1}(t) +N_{ij}^{conc,2}(t)/2= \Lambda_{ij}^{conc,1}(t) +\Lambda_{ij}^{conc,2}(t) /2 +M_{ij}^{conc,1}(t) + M_{ij}^{conc,2}(t)/2.
\end{align*}
Since the property of a process being a local martingale is closed under addition and scalar multiplication, the final two terms, which we call $M_{ij}(t)$,  form a local martingale. Therefore, the first  two terms are the compensator of $N_{ij}^{conc}(t)$, which simplify to 
\begin{align*}
    \Lambda_{ij}^{conc} &=  \int_0^t\alpha(s|Z_i)[Q_{ij}^1(s)+0.5 Q_{ij}^2(s)]Y_{ij}(s)ds.
\end{align*}

To show that $M_{ij}^{conc}$ is a martingale, and not just a local martingale, we use Theorem I.51 of \citet{Protter:2010}, which requires us to show
\begin{equation*}
    E[\sup_{s\leq t}|M_{ij}^{conc}(s)|]<\infty \quad \forall t\geq 0.
\end{equation*}
First we write
\begin{align*}
    E[\sup_{s\leq t}|M_{ij}^{conc}(s)|] &\leq E[\sup_{s\leq t}|N_{ij}^{conc}(t)|]  + E[\sup_{s\leq t}|\Lambda_{ij}^{conc}(t)|]   
    \leq 1 + E[\sup_{s\leq t}|\Lambda_{ij}^{conc}(t)|],
\end{align*}
and then we bound the second term 
\begin{align*}
    \sup_{s\leq t}|\Lambda_{ij}^{conc}(t)| &= \int_0^t(\alpha(s|Z_i)Q_{ij}^1(s) +
    \frac{\alpha(s|Z_i)Q_{ij}^2(s)}{2})Y_{ij}(s)ds\\
    &\leq \frac{3}{2}\int_0^t\alpha(s|Z_i)Y_{ij}(s)ds
    \leq \frac{3}{2}\int_0^{X_i}\alpha_i(s|Z_i)ds = \frac{3}{2}H(X_i|Z_i), 
\end{align*}
where $H(t|Z_i)=\int_0^t \alpha(s|Z_i) ds$ is the $i$th individual's integrated hazard rate.
 Suppose $Y$ is a random variable with integrated hazard rate $H$ and cumulative distribution function $F$. 
Then $E[H(Y)]=E[-\log(F(Y))]=-\int \log(F(y))dF(y)=-\int_0^1 \log(u)du=1$.
Hence, 
\begin{equation*}
    E[\sup_{s\leq t}|M_{ij}^{conc}(s)|] \leq 1 + \frac{3}{2}E[H(t\vert Z_i)]  = 5/2 <\infty.
\end{equation*}
\ 
\end{proof}

\begin{proof}[Proof of Theorem \ref{thm:2}]

For $t \in\mathbb{N}^+$ we define $N_{ij}^{comp}(t), N_{ij}^{conc,1}(t),  N_{ij}^{conc,2}(t)$ and $N_{ij}^{conc}(t)$ as in Section \ref{sec2}, with the filtration defined as 
$\mathcal{F}_t=\sigma(Z_i, \mathbb{I}(T_i\leq s), \mathbb{I}(T_i\leq s, D_i=1), i\in \mathbb{N}, s=1,\dots, t)$ and $\mathcal{F}_0 = \sigma(Z_i, i \in \mathbb{N})$
we can derive compensators of the $N_{ij}^{conc}(t)$ with respect to $\mathcal{F}_t$.

Let
\begin{equation*}
    M_{ij}(t) = N_{ij}^{conc}(t)-\Lambda_{ij}^{conc}(t)
\end{equation*}
where $\Lambda_{ij}^{conc}(t)=\sum_{s=1}^{t}Y_{ij}(s)[\mathbb{I}(q(s\vert Z_i) > q(s\vert Z_j)) + \frac{1}{2}\mathbb{I}(q(s\vert Z_i)=q(s\vert Z_j))]\alpha (s\vert Z_i)$ and $Y_{ij}(s) = \mathbb{I}(\tau_{ij} \geq s)$.

$M_{ij}(t)$ defines a discrete-time martingale with respect to $(\mathcal{F}_t)$ because
$$
E[N_{i,j}^{conc}(t)-N_{i,j}^{conc}(t-1) = 1 \vert \mathcal{F}_{t-1}, \tau_{i,j} < t] = 0
$$
and 
\begin{align*}
    &E[N_{i,j}^{conc}(t)-N_{i,j}^{conc}(t-1) = 1 \vert \mathcal{F}_{t-1}, \tau_{i,j} \geq t]= \\
    &P(\Delta N_{i,j}^{conc,1}(t) = 1 \vert \mathcal{F}_{t-1}, \tau_{i,j} \geq t) + P(\Delta N_{i,j}^{conc,2}(t) = 1 \vert \mathcal{F}_{t-1}, \tau_{i,j} \geq t)=  \\
    &\mathbb{I}(q(t\vert Z_i) > q(t\vert Z_j))\alpha(t\vert Z_i) + \mathbb{I}(q(t\vert Z_i) = q(t\vert Z_j))\alpha(t\vert Z_i) 
\end{align*}

Therefore, the compensator of $N_{ij}^{conc}(t)$ is $\Lambda_{ij}^{conc}(t)$. The proof from this point on is the same as that for Theorem \ref{thm:1} with integrals replaced by sums.

\end{proof}

\subsection{Effect of Tie Inclusion}
\label{App:Ties}
Suppose we have discrete event data $\{T_i, D_i\}_{i=1}^{n}$ for which we have potential models $M_1$ and $M_2$, with corresponding discrete hazard rates $q_1(t\vert Z_i)$ and $q_2(t \vert Z_i)$.

Let 
\begin{align*}
a &= \sum^n_{i=1}\sum^n_{j=1;j\neq i}\{\mathbb{I}[D_i=1, T_i< T_j, q_1(T_i|Z_i)>q_1(T_i|Z_j)] +\frac{1}{2}\mathbb{I}[D_i=1,T_i< T_j,q_1(T_i|Z_i)=q_1(T_i|Z_j)]\},\\
b &= \sum^n_{i=1}\sum^n_{j=1;j\neq i}\{\mathbb{I}[D_i=1, T_i< T_j, q_2(T_i|Z_i)>q_2(T_i|Z_j)] +\frac{1}{2}\mathbb{I}[D_i=1,T_i< T_j,q_2(T_i|Z_i)=q_2(T_i|Z_j)]\},\\
c &= \sum^n_{i=1}\sum^n_{j=1;j\neq i}\mathbb{I}(D_i=1,T_i< T_j)\\
\Tilde{c} &= \sum^n_{i=1}\sum^n_{j=1;j\neq i}\mathbb{I}(D_i=1,T_i= T_j),\\
w &= \frac{c}{c+2\Tilde{c}}.
\end{align*}
Suppose that $c>0$ and $\Tilde{c}>0$. With these definitions we have $c^n_{q_1}= \frac{a}{c}$, $c^n_{q_2}= \frac{b}{c}$ if ties are not included and $c^n_{q_1}=\frac{a+\Tilde{c}}{c+2\Tilde{c}}=w\frac{a}{c}+(1-w)\frac{1}{2}$, $c^n_{q_2}=\frac{b+\Tilde{c}}{c+2\Tilde{c}}=w\frac{b}{c}+(1-w)\frac{1}{2}$ if they are.  Hence ordering is the same in either case. We also have
$$
\left|\frac{a+\Tilde{c}}{c + 2\Tilde{c}}-\frac{1}{2}\right| = w\left|\frac{a}{c}-\frac{1}{2}\right| < \left|\frac{a}{c}-\frac{1}{2}\right|,
$$
showing that the inclusion of ties pulls $c_{q}^n$ closer to 0.5. Finally, the inclusion of ties pulls the estimates $c^n_{q_1}$ and $c^n_{q_2}$ closer together as
$$
\left|\frac{a+\Tilde{c}}{c + 2\Tilde{c}}-\frac{b+\Tilde{c}}{c+2\Tilde{c}}\right| = 
\left| w\frac{a}{c}+(1-w)\frac{1}{2}-w\frac{b}{c}-(1-w)\frac{1}{2}  \right|=
w\left| \frac{a}{c}-\frac{b}{c} \right|<
\left| \frac{a}{c}-\frac{b}{c}\right|.
$$

\subsection{Code}
\label{appendix:code}
Code for all experiments reported in this document can be found at \newline https://github.com/tmatcham/CrossingHazardConcordance

\bigskip

\if1\blind
{
\section*{Funding}
This article presents independent research supported by the National Institute for Health Research (NIHR) under the Applied Health Research (ARC) programme for Northwest London. The views expressed in this publication are those of the author(s) and not necessarily those of the NHS, the NIHR or the Department of Health.
TM was supported by the EPSRC Centre for Doctoral
Training in Modern Statistics and Statistical Machine Learning (EP/S023151/1).
} \fi

\section{Conflict of Interest Statement}
The authors report there are no competing interests to declare.

\bibliographystyle{chicago}

\bibliography{refs}
\end{document}